\pdfoutput=1  
\documentclass[sigconf,authorversion]{acmart}

\copyrightyear{\vspace{1cm} 2017}
\acmYear{2017}
\setcopyright{acmlicensed}
\acmConference{ISSAC '17}{July 25-28, 2017}{Kaiserslautern, Germany}
\acmPrice{15.00}
\acmDOI{10.1145/3087604.3087656}
\acmISBN{978-1-4503-5064-8/17/07}

\usepackage[english]{babel}  

\usepackage[shortlabels]{enumitem}
\usepackage{boxedminipage}

\usepackage[capitalise]{cleveref}

\makeatletter
\def\renewtheorem#1{%
  \expandafter\let\csname#1\endcsname\relax
  \expandafter\let\csname c@#1\endcsname\relax
  \gdef\renewtheorem@envname{#1}
  \renewtheorem@secpar
}
\def\renewtheorem@secpar{\@ifnextchar[{\renewtheorem@numberedlike}{\renewtheorem@nonumberedlike}}
\def\renewtheorem@numberedlike[#1]#2{\newtheorem{\renewtheorem@envname}[#1]{#2}}
\def\renewtheorem@nonumberedlike#1{  
\def\renewtheorem@caption{#1}
\edef\renewtheorem@nowithin{\noexpand\newtheorem{\renewtheorem@envname}{\renewtheorem@caption}}
\renewtheorem@thirdpar
}
\def\renewtheorem@thirdpar{\@ifnextchar[{\renewtheorem@within}{\renewtheorem@nowithin}}
\def\renewtheorem@within[#1]{\renewtheorem@nowithin[#1]}

\newtheoremstyle{framedthmenv}%
  {0cm}
  {0cm}
  {\@acmdefinitionbodyfont}
  {\@acmdefinitionindent}
  {\@acmdefinitionheadfont}
  {:}
  {.5em}
  {\thmname{#1}\thmnumber{ #2}\thmnote{ {\@acmdefinitionnotefont(#3)}}}
\makeatother

\theoremstyle{acmplain}
\renewtheorem{theorem}{Theorem}[section]
\renewtheorem{proposition}[theorem]{Proposition}
\renewtheorem{lemma}[theorem]{Lemma}
\renewtheorem{corollary}[theorem]{Corollary}
\theoremstyle{acmdefinition}
\renewtheorem{definition}[theorem]{Definition}

\theoremstyle{framedthmenv}
\newtheorem{problem}{Problem} 
\crefname{problem}{Problem}{Problems}
\Crefname{problem}{Problem}{Problems}
\newtheorem{algorithm}{Algorithm} 
\crefname{algorithm}{Algorithm}{Algorithms}
\Crefname{algorithm}{Algorithm}{Algorithms}

\theoremstyle{acmplain}

\newcommand{\myparagraph}[1]{\paragraph{\hspace{-0.35cm}\textbf{#1}}}  

\newcommand{\storeArg}{} 

\newcommand{\bigO}[1]{O(#1)} 
\newcommand{\softO}[1]{\mathchoice{\tilde{O}\left(#1\right)}{O\tilde{~}(#1)}{O\tilde{~}(#1)}{O\tilde{~}(#1)}} 
\newcommand{\expmatmul}{\omega} 

\newcommand{\algoname}[1]{{\normalfont\textsc{#1}}}
\newcommand{\problemname}[1]{{\normalfont\textsc{#1}}}
\newcommand{\algoword}[1]{\emph{\textsf{#1}}}
\newcommand{\assign}{\leftarrow}
\newcommand{\inlcomment}[1]{\texttt{\small/* #1 */}}

\renewcommand{\ge}{\geqslant} 
\renewcommand{\le}{\leqslant} 
\newcommand{\ZZ}{\mathbb{Z}} 
\newcommand{\NN}{\mathbb{Z}_{\ge 0}} 
\newcommand{\ZZp}{\mathbb{Z}_{> 0}} 
\newcommand{\tuple}[1]{\mathbf{#1}}  

\newcommand{\var}{x} 
\newcommand{\field}{\mathbb{K}} 
\newcommand{\polRing}{\field[\var]} 
\newcommand{\module}[1][M]{\mathcal{#1}} 
\newcommand{\rdim}{m} 
\newcommand{\cdim}{n} 
\newcommand{\matSpace}[1][\rdim]{\renewcommand\storeArg{#1}\matSpaceAux} 
\newcommand{\polMatSpace}[1][\rdim]{\renewcommand\storeArg{#1}\polMatSpaceAux} 
\newcommand{\matSpaceAux}[1][\storeArg]{\field^{\storeArg \times #1}} 
\newcommand{\polMatSpaceAux}[1][\storeArg]{\polRing^{\storeArg \times #1}} 
\newcommand{\idealGen}[1]{(#1)}

\newcommand{\row}[1]{\mathbf{\MakeLowercase{#1}}} 
\newcommand{\rowgrk}[1]{\boldsymbol{#1}} 
\newcommand{\mat}[1]{\mathbf{\MakeUppercase{#1}}} 
\newcommand{\matz}{\mat{0}}  
\newcommand{\sumVec}[1]{|#1|} 

\newcommand{\trsp}[1]{#1^\mathsf{T}} 
\newcommand{\matrow}[2]{{#1}_{#2,*}} 
\newcommand{\matcol}[2]{{#1}_{*,#2}} 
\newcommand{\diag}[1]{\mathrm{diag}(#1)}  
\newcommand{\idMat}[1][\rdim]{\mat{I}_{#1}} 

\newcommand{\rdeg}[2][]{\mathrm{rdeg}_{{#1}}(#2)} 
\newcommand{\cdeg}[2][]{\mathrm{cdeg}_{{#1}}(#2)} 
\newcommand{\leadingMat}[2][\unishift]{\mathrm{lm}_{#1}(#2)} 
\newcommand{\shiftSpace}[1][\rdim]{\ZZ^{#1}} 
\newcommand{\unishift}{\mathbf{0}} 
\newcommand{\shift}[2][s]{#1_{#2}} 
\newcommand{\shifts}[1][s]{\mathbf{#1}} 
\newcommand{\kerbas}{\mat{N}} 

\newcommand{\hermite}{\mat{H}} 
\newcommand{\smith}{\mat{S}} 

\newcommand{\maxDeg}{\degExp}
\newcommand{\vsdim}{D} 
\newcommand{\order}{d} 
\newcommand{\orders}{\tuple{\order}} 
\newcommand{\mulmats}{\mat{X}}
\newcommand{\minDeg}{\delta}
\newcommand{\minDegs}{\boldsymbol{\delta}}

\newcommand{\evMat}{\mat{E}} 
\newcommand{\evRow}{\row{e}} 
\newcommand{\evMatSpace}{\matSpace[\nbun][\vsdim]}
\newcommand{\mulmatSpace}{\matSpace[\vsdim]}
\newcommand{\rev}[1]{{#1}_{\mathrm{rev}}} 

\newcommand{\expandMat}{\mathcal{E}}
\newcommand{\quoExp}{\alpha}
\newcommand{\remExp}{\beta}
\newcommand{\expand}[1]{\overline{#1}}
\newcommand{\degExp}{ \delta }

\newcommand{\morphism}{\varphi_{\module,\boldsymbol{f}}}
\newcommand{\modMat}{\mat{M}}
\newcommand{\modHer}{\mat{H}}
\newcommand{\modSpace}{\polMatSpace[\nbeq]}
\newcommand{\nbeq}{\cdim} 
\newcommand{\nbun}{\rdim} 
\newcommand{\sys}{\mat{F}} 
\newcommand{\sysSpace}{\polMatSpace[\nbun][\nbeq]} 
\newcommand{\res}{\mat{G}} 
\newcommand{\rel}{\row{p}} 
\newcommand{\relSpace}{\polMatSpace[1][\nbun]} 
\newcommand{\relbas}{\mat{P}} 
\newcommand{\relbasSpace}{\polMatSpace[\nbun]} 
\newcommand{\degDet}{\vsdim_{\modMat}} 
\newcommand{\any}{\ast}
\newcommand{\anyMat}{\boldsymbol{\ast}}

\newcommand{\modRelGen}{\operatorname{\mathcal{R}}(\modMat,\sys)}  
\newcommand{\modRel}{\operatorname{\mathcal{R}}(\modHer,\sys)}     
\newcommand{\modRelCustom}[2]{\operatorname{\mathcal{R}}(#1,#2)}
\newcommand{\rem}[2]{\mathchoice{\operatorname{Rem}\left(#1,#2\right)}{\operatorname{Rem}(#1,#2)}{\operatorname{Rem}(#1,#2)}{\operatorname{Rem}(#1,#2)}}
\newcommand{\quo}[2]{\mathchoice{\operatorname{Quo}\left(#1,#2\right)}{\operatorname{Quo}(#1,#2)}{\operatorname{Quo}(#1,#2)}{\operatorname{Quo}(#1,#2)}}
\newcommand{\quoMat}{\mat{Q}}
\newcommand{\remMat}{\mat{R}}
\newcommand{\xpnt}{r}

\newenvironment{algobox}{
  \newcommand{\algoInfo}[2]{
    \begin{algorithm}
    \label{##2}
    \emph{\algoname{##1}}
  }
  \newcommand{\dataInfos}[2]{
    \algoword{##1:}
      \begin{itemize}[leftmargin=0.8cm]
          ##2
      \end{itemize}}
  \newcommand{\dataInfo}[2]{
    \algoword{##1:} ##2 }
  \newcommand{\algoSteps}[1]{
    \setlist[enumerate,1]{leftmargin=0.5cm}
    \setlist[enumerate,2]{leftmargin=0.4cm}
    \setlist[enumerate,3]{leftmargin=0.4cm}
    \begin{enumerate}[label=\textbf{\arabic*.}]
        ##1
    \end{enumerate}
  }
  \begin{figure}[ht]
  \centering
  \addtolength\fboxsep{0.1cm}
  \begin{boxedminipage}{0.99\columnwidth}
  }
  {
  \end{algorithm}
  \end{boxedminipage}
  \end{figure}
}

\newenvironment{problembox}{
  \newcommand{\problemInfo}[2]{
    \begin{problem}
    \label{##2}
    \problemname{##1}
  }
  \newcommand{\dataInfos}[2]{
    \emph{##1:}
      \begin{itemize}[leftmargin=0.8cm]
          ##2
      \end{itemize}
    }
  \newcommand{\dataInfo}[2]{
    \emph{##1:} 
      \begin{itemize}[leftmargin=0.8cm]
          \item ##2
        \end{itemize}}

  \begin{figure}[h!]
  \centering
  \addtolength\fboxsep{0.1cm}
  \begin{boxedminipage}{0.99\columnwidth}
  }
  {
  \end{problem}
  \end{boxedminipage}
  \end{figure}
}

\begin{document}
\title{Computing Canonical Bases of Modules of Univariate Relations}

\author{Vincent Neiger}
\affiliation{%
  \institution{Technical University of Denmark}
  \city{Kgs. Lyngby} 
  \state{Denmark} 
}
\email{vinn@dtu.dk}

\author{Vu Thi Xuan}
\affiliation{%
  \institution{ENS de Lyon, LIP (CNRS, Inria, ENSL, UCBL)}
  \city{Lyon} 
  \state{France} 
}
\email{thi.vu@ens-lyon.fr}

\begin{abstract}
We study the computation of canonical bases of sets of univariate relations
$(p_1,\ldots,p_m) \in \mathbb{K}[x]^{m}$ such that $p_1 f_1 + \cdots + p_m f_m
= 0$; here, the input elements $f_1,\ldots,f_m$ are from a quotient
$\mathbb{K}[x]^n/\mathcal{M}$, where $\mathcal{M}$ is a $\mathbb{K}[x]$-module
of rank $n$ given by a basis $\mathbf{M}\in\mathbb{K}[x]^{n\times n}$ in
Hermite form. We exploit the triangular shape of $\mathbf{M}$ to generalize a
divide-and-conquer approach which originates from fast minimal approximant
basis algorithms. Besides recent techniques for this approach, we rely on
high-order lifting to perform fast modular products of polynomial matrices of
the form $\mathbf{P}\mathbf{F} \bmod \mathbf{M}$.

Our algorithm uses $O\tilde{~}(m^{\omega-1}D + n^{\omega} D/m)$ operations in
$\mathbb{K}$, where $D = \deg(\det(\mathbf{M}))$ is the $\field$-vector space
dimension of $\mathbb{K}[x]^n/\mathcal{M}$, $O\tilde{~}(\cdot)$ indicates that
logarithmic factors are omitted, and $\omega$ is the exponent of matrix
multiplication. This had previously only been achieved for a diagonal matrix
$\mathbf{M}$. Furthermore, our algorithm can be used to compute the shifted
Popov form of a nonsingular matrix within the same cost bound, up to
logarithmic factors, as the previously fastest known algorithm, which is
randomized.
\end{abstract}

%
%
%


\keywords{Polynomial matrix; shifted Popov form; division with remainder;
univariate equations; syzygy module.}

\maketitle

\section{Introduction}
\label{sec:intro}

In what follows, $\field$ is a field, $\polRing$ denotes the set of univariate
polynomials in $\var$ over $\field$, and $\polRing^{\nbun\times\nbeq}$ denotes
the set of $\nbun\times\nbeq$ (univariate) polynomial matrices.

\myparagraph{Univariate relations}

Let us consider a (free) $\polRing$-submodule $\module \subseteq
\polRing^\nbeq$ of rank $\nbeq$, specified by one of its bases, represented as
the rows of a nonsingular matrix $\modMat \in \modSpace$. Besides, let some
elements $f_1,\ldots,f_\nbun \in \polRing^\nbeq / \module$ be represented as a
matrix $\sys \in \sysSpace$. Then, the kernel of the module morphism
\[
  \begin{array}{rrcl}
    \morphism: & \polRing^\nbun & \to & \polRing^\nbeq/\module \\ 
    & (p_1,\ldots,p_\nbun) & \mapsto & p_1 f_1 + \cdots + p_\nbun f_\nbun
  \end{array}
\]
consists of relations between the $f_i$'s, and is known as a \emph{syzygy
module} \cite{Eisenbud05}. From the matrix viewpoint above, we write it as
\[
  \modRelGen = \{ \rel \in \relSpace \mid \rel \sys = \matz \bmod \modMat \},
\]
where the notation $\mat{A} = \matz \bmod \modMat$ stands for ``$\mat{A} =
\mat{Q} \modMat$ for some $\mat{Q}$'', which means that the rows of $\mat{A}$
are in the module $\module$. Hereafter, the elements of $\modRelGen$ are called
\emph{relations of $\,\modRelGen$}.

Examples of such relations are the following.
\begin{itemize}
  \item \emph{Hermite-Pad\'e approximants} are relations for $\nbeq=1$ and $\module =
    \var^\vsdim \polRing$. That is, given polynomials $f_1,\ldots,f_\nbun$, the
    corresponding approximants are all $(p_1,\ldots,p_\nbun) \in
    \polRing^\nbun$ such that $p_1 f_1 + \cdots + p_\nbun f_\nbun = 0 \bmod
    \var^\vsdim$. Fast algorithms for finding such approximants include
    \cite{BarBul91,BecLab94,GiJeVi03,ZhoLab12,JeNeScVi16}.
  \item \emph{Multipoint Pad\'e approximants:} the fast computation of
    relations when $\module$ is a product of ideals, corresponding to a
    diagonal basis $\modMat=\diag{M_1,\ldots,M_\nbeq}$, was studied in
    \cite{Beckermann92,BarBul92,BecLab97,JeNeScVi17,JeNeScVi16,Neiger16}. Many
    of these references focus on $M_1,\ldots,M_\nbeq$ which split over $\field$
    with known roots and multiplicities; then, relations are known as
    multipoint Pad\'e approximants \cite{BakGraMor96}, or also
    \emph{interpolants} \cite{BecLab97,JeNeScVi17}. In this case, a relation
    can be thought of as a solution to a linear system over $\polRing$ in which
    the $j$th equation is modulo $M_j$.
\end{itemize}

\myparagraph{Canonical bases} Since $\det(\modMat) \polRing^\nbun \subseteq
\modRelGen \subseteq \polRing^\nbun$, the module $\modRelGen$ is free of rank
$\nbun$ \cite[Sec.\,12.1, Thm.\,4]{DumFoo04}. Hence, any of its bases can be
represented as the rows of a nonsingular matrix in $\relbasSpace$, which we
call a \emph{relation basis for $\modRelGen$}.

Here, we are interested in computing relation bases in \emph{shifted Popov
form} \cite{Popov72,BeLaVi99}. Such bases are canonical in terms of the module
$\modRelGen$ and of a \emph{shift}, the latter being a tuple $\shifts \in
\ZZ^\nbeq$ used as column weights in the notion of degree for row vectors.
Furthermore, the degrees in shifted Popov bases are well controlled, which
helps to compute them faster than less constrained types of bases (see
\cite{JeNeScVi16} and \cite[Sec.\,1.2.2]{Neiger16b}) and then, once obtained,
to exploit them for other purposes (see for example \cite[Thm.\,12]{RosSto16}).
Having a shifted Popov basis of a submodule $\module \subseteq \polRing^\nbeq$
is particularly useful for efficient computations in the quotient
$\polRing^\nbeq/\module$ (see \cref{sec:division}).

In fact, shifted Popov bases coincide with Gr\"obner bases for
$\polRing$-submodules of $\polRing^\nbeq$ \cite[Chap.\,15]{Eisenbud95}, for a
term-over-position monomial order weighted by the entries of the shift. For
more details about this link, we refer to \cite[Chap.\,6]{Middeke11} and
\cite[Chap.\,1]{Neiger16b}.

For a shift $\shifts = (\shift{1},\ldots,\shift{\cdim}) \in
\shiftSpace[\cdim]$, the \emph{$\shifts$-degree} of a row vector $\row{p} =
[p_1,\ldots,p_\cdim] \in \polMatSpace[1][\cdim]$ is $\max_{1\le j\le \cdim}
(\deg(p_j) + \shift{j})$; the \emph{$\shifts$-row degree} of a matrix $\mat{P}
\in \polMatSpace[\rdim][\cdim]$ is $\rdeg[\shifts]{\mat{P}} =
(d_1,\ldots,d_\rdim)$ with $d_i$ the $\shifts$-degree of the $i$th row of
$\mat{P}$. Then, the \emph{$\shifts$-leading matrix} of $\mat{P} =
[p_{i,j}]_{ij}$ is the matrix $\leadingMat[\shifts]{\mat{P}} \in
\matSpace[\rdim][\cdim]$ whose entry $(i,j)$ is the coefficient of degree $d_i
- \shift{j}$ of $p_{i,j}$. Similarly, the list of column degrees of a matrix
$\relbas$ is denoted by $\cdeg{\relbas}$.

\begin{definition}[\cite{Kailath80,BeLaVi99}]
  \label{dfn:spopov}
  Let $\relbas \in \relbasSpace$ be nonsingular, and let $\shifts \in
  \shiftSpace$. Then, $\relbas$ is said to be in
  \begin{itemize}
    \item \emph{$\shifts$-reduced form} if $\leadingMat[\shifts]{\relbas}$ is
      invertible;
    \item \emph{$\shifts$-Popov form} if $\leadingMat[\shifts]{\relbas}$
      is unit lower triangular and $\leadingMat[\unishift]{\trsp{\relbas}}$
      is the identity matrix.
  \end{itemize}
\end{definition}

\begin{problembox}
  \problemInfo
  {Relation basis}
  {pbm:relbas}

  \dataInfos{Input}{
    \item nonsingular matrix $\modMat \in \modSpace$,
    \item matrix $\sys \in \sysSpace$,
    \item shift $\shifts \in \shiftSpace$.
  }

  \dataInfo{Output}{
    the $\shifts$-Popov relation basis $\relbas \in \relbasSpace$ for $\modRelGen$.
  }
\end{problembox}

Hereafter, when we introduce a matrix by saying that it is reduced, it is
understood that it is nonsingular. Similar forms can be defined for modules
generated by the columns of a matrix rather than by its rows; in the context of
polynomial matrix division with remainder, we will use the notion of $\relbas$
in \emph{column reduced} form, meaning that
$\leadingMat[\unishift]{\trsp{\relbas}}$ is invertible. In particular, we
remark that any matrix in shifted Popov form is also column reduced.

Considering relation bases $\relbas$ for $\modRelGen$ in shifted Popov form
offers a strong control over the degrees of their entries. As shifted (row)
reduced bases, they satisfy the \emph{predictable degree property}
\cite{Forney75}, which is at the core of the correctness of a
divide-and-conquer approach behind most algorithms for the two specific
situations described above, for example
\cite{BecLab94,GiJeVi03,GioLeb14,JeNeScVi17}. Furthermore, as column reduced
matrices they have small average column degree, which is central in the
efficiency of fast algorithms for non-uniform shifts
\cite{JeNeScVi16,Neiger16}. Indeed, we will see in \cref{lem:degdetbound} that
\[
  \sumVec{\cdeg{\relbas}} = \deg(\det(\relbas)) \le \deg(\det(\mat{M})),
\]
where $\sumVec{\cdot}$ denotes the sum of the entries of a tuple.

Below, triangular canonical bases will play an important role. A matrix
$\modMat \in \modSpace$ is in \emph{Hermite form} if $\modMat$ is upper
triangular and $\leadingMat[\unishift]{\trsp{\modMat}}$ is the identity matrix;
or, equivalently, if $\modMat$ is in $(d \nbeq, d(\nbeq-1),\ldots,d)$-Popov
form for any $d \ge \deg(\det(\modMat))$.

\myparagraph{Relations modulo Hermite forms}

Our main focus is on the case where $\modMat$ is in Hermite form and $\sys$ is
already reduced modulo $\modMat$. In this article, all comparisons of tuples
are componentwise.

\begin{theorem}
  \label{thm:relbas_hermite}
  If $\modMat$ is in Hermite form and $\cdeg{\sys} < \cdeg{\modMat}$, there is
  a deterministic algorithm which solves \cref{pbm:relbas} using
  \[
    \softO{\nbun^{\expmatmul-1} \vsdim + \nbeq^\expmatmul \vsdim / \nbun}
  \]
  operations in $\field$, where $\vsdim = \deg(\det(\modMat)) =
  \sumVec{\cdeg{\modMat}}$.
\end{theorem}
\noindent
Here, the exponent $\expmatmul$ is so that we can multiply $\rdim \times \rdim$
matrices over $\field$ in $\bigO{\rdim^\expmatmul}$ operations in $\field$, the
best known bound being $\expmatmul < 2.38$~\cite{CopWin90,LeGall14}. The
notation $\softO{\cdot}$ means that we have omitted the logarithmic factors in
the asymptotic bound.

To put this cost bound in perspective, we note that the representation of the
input $\sys$ and $\modMat$ requires at most $(\nbun+\nbeq) \vsdim$ field
elements, while that of the output basis uses at most $\nbun \vsdim$ elements.
In many applications we have $\nbeq \in \bigO{\nbun}$, in which case the cost
bound becomes $\softO{\nbun^{\expmatmul-1}\vsdim}$, which is satisfactory.

To the best of our knowledge, previous algorithms with a comparable cost bound
focus on the case of a diagonal matrix $\modMat$.

The case of minimal approximant bases $\modMat = \var^\order \idMat[\nbeq]$ has
concentrated a lot of attention. A first algorithm with cost quasi-linear in
$\order$ was given \cite{BecLab94}. It was then improved in
\cite{GiJeVi03,Storjohann06,ZhoLab12}, obtaining the cost bound
$\softO{\nbun^{\expmatmul-1} \nbeq d} = \softO{\nbun^{\expmatmul-1} \vsdim}$
under assumptions on the dimensions $\nbun$ and $\nbeq$ or on the shift.

In \cite{JeNeScVi17}, the divide-and-conquer approach of \cite{BecLab94} was
carried over and made efficient in the more general case $\modMat =
\diag{M_1,\ldots,M_\nbeq}$, where the polynomials $M_i$ split over $\field$
with known linear factors. This approach was then augmented in
\cite{JeNeScVi16} with a strategy focusing on degree information to efficiently
compute the shifted Popov bases for arbitrary shifts, achieving the cost bound
$\softO{\nbun^{\expmatmul-1}\vsdim}$.

Then, the case of a diagonal matrix $\modMat$, with no assumption on the
diagonal entries, was solved within $\softO{\nbun^{\expmatmul-1} \vsdim +
\nbeq^\expmatmul \vsdim /\nbun}$ \cite{Neiger16}. The main new ingredient
developed in \cite{Neiger16} was an efficient algorithm for the case $\nbeq=1$,
that is, when solving a single linear equation modulo a polynomial; we will
also make use of this algorithm here.

In this paper we obtain the same cost bound as \cite{Neiger16} for any matrix
$\modMat$ in Hermite form. For a more detailed comparison with earlier
algorithms focusing on diagonal matrices $\modMat$, we refer the reader to
\cite[Sec.\,1.2]{Neiger16} and in particular Table\,2 therein.

Our algorithm essentially follows the approach of \cite{Neiger16}. In
particular, it uses the algorithm developed there for $\nbeq=1$. However,
working modulo Hermite forms instead of diagonal matrices makes the computation
of \emph{residuals} much more involved. The residual is a modular product
$\relbas \sys \bmod \modMat$ which is computed after the first recursive call
and is to be used as an input replacing $\sys$ for the second recursive call.
When $\modMat$ is diagonal, its computation boils down to the multiplication of
$\relbas$ and $\sys$, although care has to be taken to account for their
possibly unbalanced column degrees. However, when $\modMat$ is triangular,
computing $\relbas \sys \bmod \modMat$ becomes a much greater challenge: we
want to compute a matrix remainder instead of simply taking polynomial
remainders for each column separately. We handle this, while still taking
unbalanced degrees into account, by resorting to high-order lifting
\cite{Storjohann03}.

\myparagraph{Shifted Popov forms of matrices}

A specific instance of \cref{pbm:relbas} yields the following problem: given a
shift $\shifts \in \ZZ^{\nbeq}$ and a nonsingular matrix $\mat{M} \in
\polMatSpace[\nbeq]$, compute the $\shifts$-Popov form of $\mat{M}$. Indeed,
the latter is the $\shifts$-Popov relation basis for
$\modRelCustom{\mat{M}}{\idMat[\nbeq]}$ (see \cref{lem:relbas_popovform}). 

To compute this relation basis efficiently, we start by computing the Hermite
form $\hermite$ of $\modMat$, which can be done deterministically in
$\softO{\nbeq^{\expmatmul} \lceil\degDet/\nbeq\rceil}$ operations
\cite{LaNeZh17}. Here, $\degDet$ is the \emph{generic determinant bound}
\cite{GuSaStVa12}; writing $\modMat = [a_{ij}]$, it is defined as
\[
  {\textstyle
    \degDet = \max_{\pi \in S_\nbeq} \sum_{1 \le i\le \nbeq} \max(0,\deg(a_{i,\pi_i}))
  }
\]
where $S_\nbeq$ is the set of permutations of $\{1,\ldots,\nbeq\}$. In
particular, $\degDet/\nbeq$ is bounded from above by both the average of the
degrees of the columns of $\modMat$ and that of its rows. For more details
about this quantity, we refer to \cite[Sec.\,6]{GuSaStVa12} and
\cite[Sec.\,2.3]{LaNeZh17}.

Since the rows of $\hermite$ generate the same module as $\modMat$, we have
$\modRelCustom{\modMat}{\idMat[\nbeq]} = \modRelCustom{\modHer}{\idMat[\nbeq]}$
(see \cref{lem:rel_operations}). Then, applying our algorithm for relations
modulo $\modHer$ has a cost of $\softO{\nbeq^{\expmatmul-1}
\deg(\det(\modHer))}$ operations, according to \cref{thm:relbas_hermite}. This
yields the next result.

\begin{theorem}
  \label{thm:spopov}
  Given a shift $\shifts \in \ZZ^\nbeq$ and a nonsingular matrix $\modMat \in
  \modSpace$, there is a deterministic algorithm which computes the
  $\shifts$-Popov form of $\modMat$ using
  \[
    \softO{\nbeq^\expmatmul \lceil\degDet/\nbeq\rceil} \;\subseteq\;
    \softO{\nbeq^\expmatmul \deg(\modMat)}
  \]
  operations in $\field$.
\end{theorem}
\noindent
A similar cost bound was obtained in \cite{Neiger16}, yet with a randomized
algorithm. The latter follows the approach of \cite{GupSto11} for computing
Hermite forms, whose first step determines the Smith form $\smith$ of $\modMat$
along with a matrix $\sys$ such that the sought matrix is the $\shifts$-Popov
relation basis for $\modRelCustom{\smith}{\sys}$, with $\smith$ being therefore
a diagonal matrix. Here, relying on the deterministic computation of the
Hermite form of $\modMat$, our algorithm for relation bases modulo Hermite
forms allows us to circumvent the computation of $\smith$, for which the
currently fastest known algorithm is Las Vegas randomized \cite{Storjohann03}.
For a more detailed comparison with earlier row reduction and Popov forms
algorithms, we refer to \cite[Sec.\,1.1]{Neiger16} and Table\,1 therein.

\myparagraph{General relation bases}

To solve the general case of \cref{pbm:relbas}, one can proceed as follows:
\begin{itemize}
  \item find the Hermite form $\modHer$
    of $\modMat$, using \cite[Algo.\,1 and\,3]{LaNeZh17};
  \item reduce $\sys$ modulo $\modHer$, for example using
    \cref{algo:division};
  \item apply \cref{algo:relbas} for relations modulo a Hermite form.
\end{itemize}

\myparagraph{Outline}

We first give basic properties about matrix division and relation bases
(\cref{sec:preliminaries}). We then focus on the fast computation of residuals
(\cref{sec:division}). After that, we discuss three situations which have
already been solved efficiently in the literature (\cref{sec:building_blocks}):
when $\nbeq=1$, when information on the output degrees is available, and when
$\vsdim \le \nbun$. Finally, we present our algorithm for relations modulo
Hermite forms (\cref{sec:relbas}).

\section{Preliminaries on polynomial matrix division and modules of relations}
\label{sec:preliminaries}

\myparagraph{Division with remainder}

Polynomial matrix division is a central notion in this paper, since we aim at
solving equations modulo $\modMat$.

\begin{theorem}[{\cite[IV.\S2]{Gantmacher59},\cite[Thm.\,6.3-15]{Kailath80}}]
  \label{thm:quo_rem}
  For any $\sys \in \sysSpace$ and any column reduced $\modMat \in \modSpace$,
  there exist unique matrices $\quoMat,\remMat \in \sysSpace$ such that $\sys = \quoMat
  \modMat + \remMat$ and $\cdeg{\remMat} < \cdeg{\modMat}$.
\end{theorem}

Hereafter, we write $\quo{\sys}{\modMat}$ and $\rem{\sys}{\modMat}$ for the
quotient $\quoMat$ and the remainder $\remMat$. We have the following
properties.

\begin{lemma}
  \label{lem:quo_rem_properties}
  We have $\rem{\relbas\rem{\sys}{\modMat}}{\modMat} =
  \rem{\relbas\sys}{\modMat}$ and $\displaystyle\rem{\begin{bmatrix} \sys \\
    \res
    \end{bmatrix}}{\modMat} =
    \begin{bmatrix}
      \rem{\sys}{\modMat} \\ \rem{\res}{\modMat}
    \end{bmatrix}$
  for any $\sys \in \sysSpace$, $\res \in \polMatSpace[\any][\nbeq]$, $\relbas
  \in \polMatSpace[\any][\nbun]$ and any column reduced $\modMat \in
  \modSpace$.
\end{lemma}

\myparagraph{Degree control for relation bases} We first relate the vector
space dimension of quotients and the degree of determinant of bases.

\begin{lemma}
  \label{lem:vsdim_colredbasis}
  Let $\module$ be a $\polRing$-submodule of $\polRing^\nbeq$ of rank $\nbeq$.
  Then, the dimension of $\polRing^{\nbeq}/\module$ as a $\field$-vector space
  is $\deg(\det(\modMat))$, for any matrix $\modMat \in \modSpace$ whose rows
  form a basis of $\module$.
\end{lemma}
\begin{proof}
  Since the degree of the determinant is the same for all bases of $\module$,
  we may assume that $\modMat$ is column reduced. Then, \cref{thm:quo_rem}
  implies that there is a $\field$-vector space isomorphism
  $\polRing^\nbeq/\module \cong \polRing/\idealGen{\var^{\order_1}} \times
  \cdots \times \polRing/\idealGen{\var^{\order_\nbeq}}$, where $(\order_1,
  \ldots, \order_\nbeq) = \cdeg{\modMat}$. Thus, the dimension of
  $\polRing^{\nbeq}/\module$ is $\order_1 + \cdots + \order_\nbeq$, which is
  equal to $\deg(\det(\modMat))$ according to \cite[Sec.\,6.3.2]{Kailath80}.
\end{proof}
%
%

This allows us to bound the sum of column degrees of any column reduced
relation basis; for example, a shifted Popov relation basis.

\begin{corollary}
  \label{lem:degdetbound}
  Let $\sys \in \sysSpace$, and let $\modMat \in \modSpace$ be nonsingular.
  Then, any relation basis $\relbas \in \polMatSpace$ for $\modRelGen$ is such
  that $\deg(\det(\relbas)) \le \deg(\det(\modMat))$. In particular, if
  $\relbas$ is column reduced, then $\sumVec{\cdeg{\relbas}} \le
  \deg(\det(\modMat))$.
\end{corollary}
\begin{proof}
  Let $\module$ be the row space of $\modMat$. By definition, $\modRelGen$ is
  the kernel of $\morphism$ (see \cref{sec:intro}), hence $\polRing^\nbun /
  \modRelGen$ is isomorphic to a submodule of $\polRing^\nbeq / \module$.
  Since, by \cref{lem:vsdim_colredbasis}, the dimensions of $\polRing^\nbun /
  \modRelGen$ and $\polRing^\nbun / \module$ are $\deg(\det(\relbas))$ and
  $\deg(\det(\modMat))$, we obtain $\deg(\det(\relbas)) \le
  \deg(\det(\modMat))$.
\end{proof}

\myparagraph{Properties of relation bases}

We now formalize the facts that $\modRelGen$ is not changed if $\modMat$ is
replaced by another basis of the module generated by its rows; or if $\sys$ and
$\modMat$ are right-multiplied by the same nonsingular matrix; or yet if $\sys$
is considered modulo $\modMat$.

\begin{lemma}
  \label{lem:rel_operations}
  Let $\sys \in \sysSpace$, and let $\modMat \in \modSpace$ be nonsingular.
  Then, for any nonsingular $\mat{A} \in \modSpace$, any matrix $\mat{B} \in
  \sysSpace$, and any unimodular $\mat{U} \in \relbasSpace$, we have
  \[
    \modRelGen =
    \modRelCustom{\mat{U}\modMat}{\sys} =
    \modRelCustom{\modMat\mat{A}}{\sys\mat{A}} =
    \modRelCustom{\modMat}{\sys+\mat{B}\modMat}.
  \]
\end{lemma}
\noindent
A first consequence is that we may discard identity columns in $\modMat$.

\begin{corollary}
  \label{cor:cleaning}
  Let $\sys \in \sysSpace$, and let $\modMat \in \modSpace$ be nonsingular.
  Suppose that $\modMat$ has at least $k\in\ZZp$ identity columns, and that the
  corresponding columns of $\sys$ are zero. Then, let $\pi_1,\pi_2$
  be $\nbeq\times\nbeq$ permutation matrices such that
  \[
    \pi_1 \modMat \pi_2 = \begin{bmatrix} \idMat[k] & \mat{B} \\ \matz & \mat{N} \end{bmatrix}
    \;\text{ and }\;\;
    \sys \pi_2 = \begin{bmatrix} \matz & \mat{G} \end{bmatrix},
  \]
  where $\mat{G} \in \polMatSpace[\nbun][(\nbeq-k)]$. Then, $\modRelGen =
  \modRelCustom{\mat{N}}{\mat{G}}$.
\end{corollary}

Another consequence concerns the transformation of a matrix into shifted Popov
form. Indeed, \cref{lem:rel_operations} together with the next lemma imply in
particular that the $\shifts$-Popov form of $\mat{M}$ is the $\shifts$-Popov
relation basis for $\modRelCustom{\hermite}{\idMat[\nbeq]}$, where $\hermite$
is the Hermite form of $\mat{M}$.

\begin{lemma}
  \label{lem:relbas_popovform}
  Let $\mat{M} \in \modSpace$ be nonsingular. Then, $\mat{M}$ is a relation
  basis for $\modRelCustom{\mat{M}}{\idMat[\nbeq]}$. It follows that the
  $\shifts$-Popov form of $\mat{M}$ is the $\shifts$-Popov relation basis for
  $\modRelCustom{\mat{M}}{\idMat[\nbun]}$, for any $\shifts \in \ZZ^\nbeq$.
\end{lemma}
\begin{proof}
  Let $\relbas \in \modSpace$ be a relation basis for
  $\modRelCustom{\mat{M}}{\idMat[\nbeq]}$. Then, $\relbas \idMat[\nbeq] =
  \mat{Q}\mat{M}$ for some $\mat{Q} \in \modSpace$; since the rows of $\mat{M}$
  belong to $\modRelCustom{\mat{M}}{\idMat[\nbeq]}$, we also have $\mat{M} =
  \mat{R}\relbas$ for some $\mat{R} \in \modSpace$. Since $\relbas$ is
  nonsingular, $\relbas = \mat{Q}\mat{R}\relbas$ implies that $\mat{Q}\mat{R} =
  {\idMat[\nbeq]}$, and therefore $\mat{R}$ is unimodular. Thus, $\mat{M} =
  \mat{R}\relbas$ is a relation basis for
  $\modRelCustom{\mat{M}}{\idMat[\nbeq]}$.
\end{proof}

\myparagraph{Divide and conquer approach}

Here we give properties in the case of a block triangular matrix $\modMat$.
They imply, if $\modMat$ is in Hermite form, that \cref{pbm:relbas} can be
solved recursively by splitting the instance in dimension $\nbeq$ into two
instances in dimension $\nbeq/2$.

\begin{lemma}
  \label{lem:triangular_rem}
  Let $\modMat_1\in\polMatSpace[\nbeq_1]$, $\modMat_2\in\polMatSpace[\nbeq_2]$,
  and $\mat{A}\in\polMatSpace[\nbeq_1][\nbeq_2]$ be such that
  $\modMat=\big[\begin{smallmatrix} \modMat_1 & \mat{A} \\ \matz & \modMat_2
  \end{smallmatrix}\big]$ is column reduced. For any $\sys_1 \in
  \polMatSpace[\nbun][\nbeq_1]$ and $\sys_2 \in \polMatSpace[\nbun][\nbeq_2]$,
  we have $\rem{[\sys_1\;\;\sys_2]}{\modMat} = [ \rem{\sys_1}{\modMat_1} \;\;
  \rem{\sys_2 - \quo{\sys_1}{\modMat_1}\mat{A}}{\modMat_2}]$.
\end{lemma}
\begin{proof}
  Writing $[\sys_1\;\;\sys_2] = [\quoMat_1\;\;\quoMat_2]\modMat +
  [\remMat_1\;\;\remMat_2]$ where
  $\cdeg{[\remMat_1\;\;\remMat_2]}<\cdeg{\modMat}$, we obtain $\sys_1 =
  \quoMat_1 \modMat_1 + \remMat_1$ as well as
  $\cdeg{\remMat_1}<\cdeg{\modMat_1}$, and therefore $\remMat_1 =
  \rem{\sys_1}{\modMat_1}$ and $\quoMat_1=\quo{\sys_1}{\modMat_1}$. The result
  follows from $\sys_2 = \quoMat_1 \mat{A} + \quoMat_2 \modMat_2 +
  \remMat_2$.
\end{proof}

\begin{theorem}
  \label{thm:basis_transitivity}
  Let $\modMat=\big[\begin{smallmatrix} \modMat_1 & \anyMat \\ \matz &
    \modMat_2 \end{smallmatrix}\big]$ be column reduced, where
  $\modMat_1\in\polMatSpace[\nbeq_1]$ and $\modMat_2\in\polMatSpace[\nbeq_2]$,
  and let $\sys_1 \in \polMatSpace[\nbun][\nbeq_1]$ and $\sys_2 \in
  \polMatSpace[\nbun][\nbeq_2]$. If $\relbas_1$ is a basis for
  $\modRelCustom{\modMat_1}{\sys_1}$, then
  $\rem{\relbas_1[\sys_1\;\;\sys_2]}{\modMat}$ has the form $[\matz\;\;\res]$
  for some $\res\in\polMatSpace[\nbun][\nbeq_2]$; if furthermore $\relbas_2$ is
  a basis for $\modRelCustom{\modMat_2}{\res}$, then $\relbas_2\relbas_1$ is a
  basis for $\modRelCustom{\modMat}{[\sys_1\;\;\sys_2]}$.
\end{theorem}
\begin{proof}
  It follows from \cref{lem:triangular_rem} that the first $\nbeq_1$ columns of
  $\rem{\relbas_1[\sys_1\;\;\sys_2]}{\modMat}$ are
  $\rem{\relbas_1\sys_1}{\modMat_1}$, which is zero, and that
  $\rem{[\matz\;\;\res]}{\modMat} = [\matz\;\;\rem{\res}{\modMat_2}]$. Then,
  the first identity in \cref{lem:quo_rem_properties} implies both that
  $\modRelCustom{\modMat}{[\matz\;\;\res]} = \modRelCustom{\modMat_2}{\res}$
  and that the rows of $\relbas_2\relbas_1$ are in
  $\modRelCustom{\modMat}{[\sys_1\;\;\sys_2]}$. Now let
  $\rel\in\modRelCustom{\modMat}{[\sys_1\;\;\sys_2]}$.
  \cref{lem:triangular_rem} implies that
  $\rel\in\modRelCustom{\modMat_1}{\sys_1}$, hence $\rel =
  \rowgrk{\lambda}\relbas_1$ for some $\rowgrk{\lambda}$. Then, the first
  identity in \cref{lem:quo_rem_properties} shows that $\matz =
  \rem{\rowgrk{\lambda}\relbas_1[\sys_1\;\;\sys_2]}{\modMat} =
  \rem{\rowgrk{\lambda}[\matz\;\;\res]}{\modMat}$, and therefore
  $\rowgrk{\lambda} \in \modRelCustom{\modMat_2}{\res}$. Thus $\rowgrk{\lambda}
  = \rowgrk{\mu}\relbas_2$ for some $\rowgrk{\mu}$, and
  $\rel=\rowgrk{\mu}\relbas_2\relbas_1$.
\end{proof}

\section{Computing modular products}
\label{sec:division}

In this section, we aim at designing a fast algorithm for the modular products
that arise in our relation basis algorithm.

\subsection{Fast division with remainder}
\label{subsec:division:uniform}

For univariate polynomials, fast Euclidean division can be achieved by first
computing the reversed quotient via Newton iteration, and then deducing the
remainder \cite[Chap.\,9]{vzGathen13}. This directly translates into the
context of polynomial matrices, as was noted for example in the proof of
\cite[Lem.\,3.4]{GiJeVi03} or in \cite[Chap.\,10]{Zhou12}.

In the latter reference, it is showed how to efficiently compute remainders
$\rem{\expandMat}{\modMat}$ for a matrix $\expandMat$ as in
\cref{eqn:expandMat} below; this is not general enough for our purpose.
Algorithms for the general case have been studied
\cite{FavLot91,ZhaChe83,Wolovich84,CodLot89,WanZho86}, but we are not aware of
any that achieves the speed we desire. Thus, as a preliminary to the
computation of residuals in \cref{subsec:division:non_uniform}, we now detail
this extension of fast polynomial division to fast polynomial matrix division.

As mentioned above, we will start by computing the quotient. The degrees of its
entries are controlled thanks to the reducedness of the divisor, which ensures
that no high-degree cancellation can occur when multiplying the quotient and
the divisor.

\begin{lemma}
  \label{lem:quotient_degree}
  Let $\modMat \in \modSpace$, $\sys \in \sysSpace$, and $\maxDeg \in \ZZp$ be
  such that $\modMat$ is column reduced and $\cdeg{\sys} < \cdeg{\modMat} +
  (\maxDeg,\ldots,\maxDeg)$. Then, $\deg(\quo{\sys}{\modMat}) < \maxDeg$.
\end{lemma}
\begin{proof}
  First, $\trsp{\leadingMat[\unishift]{\trsp{\modMat}}} =
  \leadingMat[{-\orders}]{\modMat}$ where $\orders=\cdeg{\modMat}\in\NN^\nbeq$:
  the $\unishift$-column leading matrix of $\modMat$ is equal to its
  $-\orders$-row leading matrix. Since $\modMat$ is $\unishift$-column reduced,
  it is also $-\orders$-row reduced.

  Thus, by the predictable degree property \cite[Thm.\,6.3-13]{Kailath80} and
  since since $\rdeg[-\orders]{\modMat} = \unishift$, we have
  $\rdeg[-\orders]{\quoMat \modMat} = \rdeg[\unishift]{\quoMat}$. Here, we
  write $\quoMat = \quo{\sys}{\modMat}$ and $\remMat=\rem{\sys}{\modMat}$.

  Now, our assumption $\cdeg{\sys} < \orders + (\maxDeg,\ldots,\maxDeg)$ and
  the fact that $\cdeg{\remMat} < \orders$ imply that $\cdeg{\sys - \remMat} <
  \orders + (d,\ldots,d)$, and thus $\rdeg[-\orders]{\sys-\remMat} <
  (\maxDeg,\ldots,\maxDeg)$. Since $\sys - \remMat = \quoMat \modMat$, from the
  previous paragraph we obtain $\rdeg[\unishift]{\quoMat} <
  (\maxDeg,\ldots,\maxDeg)$, hence $\deg(\quoMat) < \maxDeg$.
\end{proof}

\begin{corollary}
  \label{cor:quotient_degree}
  Let $\modMat \in \modSpace$ and $\sys \in \sysSpace$ be such that $\modMat$
  is column reduced and $\cdeg{\sys} < \cdeg{\modMat}$, and let $\mat{P} \in
  \polMatSpace[k][\nbun]$. Then, $\rdeg{\quo{\mat{P} \sys}{\modMat}} <
  \rdeg{\mat{P}}$.
\end{corollary}
\begin{proof}
  For the case $k=1$, the inequality follows from \cref{lem:quotient_degree}
  since $\cdeg{\mat{P} \sys} \le (\maxDeg,\ldots,\maxDeg) + \cdeg{\sys} <
  (\maxDeg,\ldots,\maxDeg) + \cdeg{\modMat}$, where $\maxDeg = \deg(\mat{P})$.
  Then, the general case $k \in \ZZp$ follows by considering separately each
  row of $\mat{P}$.
\end{proof}

Going back to the division $\sys = \quoMat \modMat + \remMat$, to obtain the
reversed quotient we will right-multiply the reversed $\sys$ by an expansion of
the inverse of the reversed $\modMat$. This operation is performed efficiently
by means of high-order lifting; we will use the next result.

\begin{lemma}
  \label{lem:hol}
  Let $\modMat \in \modSpace$ with $\modMat(0)$ nonsingular, and let $\sys \in
  \sysSpace$. Then, defining $\order = \lceil \sumVec{\cdeg{\modMat}} / \nbeq
  \rceil$, the truncated $\var$-adic expansion $\sys \modMat^{-1} \bmod \var^{k
  \order}$ can be computed deterministically using $\softO{\lceil \nbun k /
    \nbeq \rceil \nbeq^\expmatmul \order}$ operations in $\field$.
\end{lemma}
\begin{proof}
  This is a minor extension of \cite[Prop.\,15]{Storjohann03}, incorporating the
  average column degree of the matrix $\modMat$ instead of the largest degree of
  its entries. This can be done by means of partial column linearization
  \cite[Sec.\,6]{GuSaStVa12}, as follows. One first expands the high-degree
  columns of $\modMat$ and inserts elementary rows to obtain a matrix
  $\expand{\modMat} \in \polMatSpace[\expand{\nbeq}]$ such that $\nbeq \le
  \expand{\nbeq} < 2\nbeq$, $\deg(\expand{\modMat}) \le \order$, and
  $\modMat^{-1}$ is the $\nbeq\times\nbeq$ principal leading submatrix of
  $\expand{\modMat}{ }^{-1}$ \cite[Thm.\,10 and Cor.\,2]{GuSaStVa12}. Then,
  defining $\expand{\sys} = [\sys \;\; \mat{0}] \in
  \polMatSpace[\nbun][\expand{\nbeq}]$, we have that $\sys \modMat^{-1}$ is the
  submatrix of $\expand{\sys} \, \expand{\modMat}{ }^{-1}$ formed by its first
  $\nbeq$ columns. Thus, the sought truncated expansion is obtained by computing
  $\expand{\sys} \, \expand{\modMat}{ }^{-1} \bmod \var^{k \order}$, which is
  done efficiently by \cite[Alg.\,4]{Storjohann03} with the choice $X =
  \var^\order$; this is valid since this polynomial is coprime to
  $\det(\expand{\modMat}) = \det(\modMat)$ and its degree is at least the degree
  of $\expand{\modMat}$.
\end{proof}

\begin{algobox}
  \algoInfo
  {PM-QuoRem}
  {algo:division}

  \dataInfos{Input}{
    \item $\modMat \in \modSpace$ column reduced,
    \item $\sys \in \sysSpace$,
    \item $\maxDeg \in \ZZp$ such that $\cdeg{\sys} < \cdeg{\modMat} + (\maxDeg,\ldots,\maxDeg)$.
  }

  \dataInfo{Output}{
    the quotient $\quo{\sys}{\modMat}$, the remainder $\rem{\sys}{\modMat}$.
  }

  \algoSteps{
    \item \inlcomment{reverse order of coefficients} \\
      $(\order_1,\ldots,\order_\nbeq) \assign \cdeg{\modMat}$ \\
      $\rev{\modMat} = \modMat(\var^{-1}) \, \diag{\var^{\order_1},\ldots,\var^{\order_\nbeq}}$ \\
      $\rev{\sys} = \sys(\var^{-1}) \, \diag{\var^{\maxDeg+\order_1-1},\ldots,\var^{\maxDeg+\order_\nbeq-1}}$
    \item \inlcomment{compute quotient via expansion} \\
      $\rev{\quoMat} \assign \rev{\sys} \rev{\modMat}^{-1} \bmod \var^{\maxDeg}$ \\
      $\quoMat \assign \var^{\maxDeg-1}\rev{\quoMat}(\var^{-1})$
    \item \algoword{Return} $(\quoMat,\sys - \quoMat \modMat)$
  }
\end{algobox}

\begin{proposition}
  \label{prop:algo:quo_rem}
  \cref{algo:division} is correct. Assuming that both $\nbun\degExp$ and
  $\nbeq$ are in $\bigO{\vsdim}$, where $\vsdim = \sumVec{\cdeg{\modMat}}$,
  this algorithm uses $\softO{\lceil\nbun/\nbeq\rceil \nbeq^{\expmatmul-1}
\vsdim}$ operations in $\field$.
\end{proposition}
\begin{proof}
  Let $\quoMat=\quo{\sys}{\modMat}$, $\remMat=\rem{\sys}{\modMat}$, and
  $(\order_1,\ldots,\order_\nbeq)=\cdeg{\modMat}$. We have the bounds
  $\cdeg{\sys} < (\maxDeg+\order_1,\ldots,\maxDeg+\order_\nbeq)$,
  $\cdeg{\remMat} < (\order_1,\ldots,\order_\nbeq)$, and
  \cref{lem:quotient_degree} gives $\deg(\quoMat) < \maxDeg$. Thus, we can
  define the reversals of these polynomial matrices as
  \begin{align*}
      \rev{\modMat} & = \modMat(\var^{-1}) \, \diag{\var^{\order_1},\ldots,\var^{\order_\nbeq}}, \\
      \rev{\sys} & = \sys(\var^{-1}) \, \diag{\var^{\maxDeg+\order_1-1},\ldots,\var^{\maxDeg+\order_\nbeq-1}}, \\
      \rev{\quoMat} & = \var^{\maxDeg-1} \quoMat(\var^{-1}), \\
      \rev{\remMat} & = \remMat(\var^{-1}) \, \diag{\var^{\order_1-1},\ldots,\var^{\order_\nbeq-1}},
  \end{align*}
  for which the same degree bounds hold. Then, right-multiplying both sides of
  the identity $\sys(\var^{-1}) = \quoMat(\var^{-1}) \modMat(\var^{-1}) +
  \remMat(\var^{-1})$ by
  $\diag{\var^{\maxDeg+\order_1-1},\ldots,\var^{\maxDeg+\order_\nbeq-1}}$, we
  obtain $\rev{\sys} = \rev{\quoMat} \rev{\modMat} + \var^\maxDeg
  \rev{\remMat}$.

  Now, note that the constant term $\rev{\modMat}(0) \in \matSpace[\nbeq]$ is
  equal to the column leading matrix of $\modMat$, which is invertible since
  $\modMat$ is column reduced, hence $\rev{\modMat}$ is invertible (over the
  fractions). Thus, since $\deg(\rev{\quoMat}) < \maxDeg$, this reversed
  quotient matrix can be determined as the truncated expansion $\rev{\quoMat} =
  \rev{\sys} \rev{\modMat}^{-1} \bmod \var^{\maxDeg}$. This proves the
  correctness of the algorithm.

  Concerning the cost bound, Step\,\textbf{2} uses $\softO{\lceil
  (\nbun\maxDeg) / (\nbeq \order)\rceil \nbeq^\expmatmul \order}$ operations
  according to \cref{lem:hol}, where $\order = \lceil\vsdim /\nbeq\rceil$. We
  have by assumption $\order \in \Theta(\vsdim / \nbeq)$ as well as
  $\nbun\maxDeg/(\nbeq\order) \in \bigO{1}$, so that this cost bound is in
  $\softO{\nbeq^{\expmatmul-1} \vsdim}$.

  In Step\,\textbf{3}, we multiply the $\nbun\times\nbeq$ matrix $\quoMat$ of
  degree less than $\maxDeg$ with the $\nbeq\times\nbeq$ matrix $\modMat$ such
  that $\sumVec{\cdeg{\modMat}} = \vsdim$. First consider the case $\nbun \le
  \nbeq$. To perform this product efficiently, we expand the rows of $\quoMat$
  so as to obtain a $\bigO{\nbeq}\times\nbeq$ matrix $\expand{\quoMat}$ of
  degree in $\bigO{\lceil \nbun\degExp / \nbeq\rceil}$ and such that $\quoMat
  \modMat$ is easily retrieved from $\expand{\quoMat} \modMat$ (see
  \cref{subsec:division:non_uniform} for more details about how such row
  expansions are carried out). Thus, this product is done in
  $\softO{\nbeq^{\expmatmul-1} \vsdim}$, since $\lceil \nbun\degExp /
  \nbeq\rceil \in \bigO{\vsdim/\nbeq}$. On the other hand, if $\nbun > \nbeq$,
  we have $\degExp \in \bigO{\vsdim/\nbun} \subseteq \bigO{\vsdim/\nbeq}$.
  Then, we can compute the product $\quoMat \modMat$ via
  $\lceil\nbun/\nbeq\rceil$ products of $\nbeq\times \nbeq$ matrices of degree
  $\bigO{\vsdim/\nbeq}$, which cost each $\softO{\nbeq^{\expmatmul-1} \vsdim}$
  operations; hence the total cost $\softO{\nbun \nbeq^{\expmatmul-2} \vsdim}$
  when $\nbun > \nbeq$.
\end{proof}

\subsection{Fast residual computation}
\label{subsec:division:non_uniform}

Here, we focus on performing modular products $\rem{\relbas\sys}{\modMat}$,
where $\sys \in \sysSpace$ and $\relbas \in \relbasSpace$ are such that
$\cdeg{\sys}<\cdeg{\modMat}$ and $\sumVec{\cdeg{\relbas}} \le
\sumVec{\cdeg{\modMat}}$, and $\modMat \in \modSpace$ is column reduced. The
difficulty in designing a fast algorithm for this operation comes from the
non-uniformity of $\cdeg{\relbas}$: in particular, the product $\relbas \sys$
cannot be computed within the target cost bound.

To start with, we use the same strategy as in \cite{JeNeScVi16,Neiger16}: we
make the column degrees of $\relbas$ uniform, at the price of introducing
another, simpler matrix $\expandMat$ for which we want to compute
$\rem{\expandMat \sys}{\modMat}$.

Let $(\minDeg_1,\ldots,\minDeg_\nbun) = \cdeg{\relbas}$, $\degExp =
\lceil(\minDeg_1 + \cdots + \minDeg_\nbun) / \nbun \rceil \ge 1$, and for
$i\in\{1,\ldots,\nbun\}$ write $\minDeg_i = (\quoExp_i -1) \degExp + \remExp_i$
with $\quoExp_i = \lceil \minDeg_i / \degExp \rceil$ and $1 \le \remExp_i \le
\degExp$ if $\minDeg_i > 0$, and with $\quoExp_i = 1$ and $\remExp_i = 0$ if
$\minDeg_i = 0$. Then, let $\expand{\nbun} = \quoExp_1 + \cdots +
\quoExp_\nbun$, and define $\expandMat \in
\polMatSpace[\expand{\nbun}][\nbun]$ as the transpose of
\begin{equation}
  \label{eqn:expandMat}
  \trsp{\expandMat} = 
  \begin{bmatrix}
    1\! & \!\var^\degExp\! & \!\cdots\! & \!\var^{(\quoExp_1-1)\degExp} \\
      & & & & \!\ddots\! \\
      & & & &        & \!1\! & \!\var^\degExp\! & \!\cdots\! & \!\var^{(\quoExp_\nbun-1)\degExp}
  \end{bmatrix}.
\end{equation}
Define also the expanded column degrees $\expand{\minDegs}\in
\NN^{\expand{\nbun}}$ as
\begin{equation}
  \label{eqn:expandMinDegs}
  \expand{\minDegs} = ( \underbrace{\degExp, \ldots, \degExp,
  \remExp_1}_{\quoExp_1}, \ldots, \underbrace{\degExp, \ldots, \degExp,
  \remExp_\nbun}_{\quoExp_\nbun} ).
\end{equation}
\noindent
Then, we expand the columns of $\relbas$ by considering $\expand{\relbas} \in
\polMatSpace[\nbun][\expand{\nbun}]$ such that $\relbas = \expand{\relbas}
\expandMat$ and $\deg(\expand{\relbas}) \le \degExp$. (Note that
$\expand{\relbas}$ can be made unique by specifying more constraints on
$\cdeg{\expand{\relbas}}$.) The aim of this construction is that the dimension
is at most doubled while the degree of the expanded matrix becomes the average
column degree of $\relbas$.  Precisely, $\nbun \le \expand{\nbun} < 2\nbun$ and
$\max(\expand{\minDegs}) = \degExp = \lceil
\sumVec{\cdeg{\relbas}}/\nbun\rceil$.

Now, we have $\rem{\relbas\sys}{\modMat} =
\rem{\expand{\relbas}\expandMat\sys}{\modMat} =
\rem{\expand{\relbas}\,\expand{\sys}}{\modMat}$ by
\cref{lem:quo_rem_properties}, where $\expand{\sys} =
\rem{\expandMat\sys}{\modMat}$. Thus, $\rem{\relbas\sys}{\modMat}$ can be
obtained by computing first $\expand{\sys}$ and then
$\rem{\expand{\relbas}\,\expand{\sys}}{\modMat}$. For the latter, since
$\expand{\relbas}$ has small degree, one can compute the product and then
perform the division (Steps\,\textbf{3} and \textbf{4} of
\cref{algo:residual}). Step\,\textbf{2} of \cref{algo:residual} efficiently
computes $\expand{\sys}$, relying on \cref{algo:rem_of_shifts}.

\begin{algobox}
  \algoInfo
  {RemOfShifts}
  {algo:rem_of_shifts}

  \dataInfos{Input}{
    \item $\modMat \in \modSpace$ column reduced,
    \item $\sys \in \sysSpace$ such that $\cdeg{\sys}<\cdeg{\modMat}$,
    \item $\degExp \in \ZZp$ and $k \in \NN$.
  }

  \dataInfo{Output}{
    the list of remainders $(\rem{\var^{\xpnt \degExp} \sys}{\modMat})_{0 \le \xpnt<2^k}$.
  }

  \algoSteps{
    \item \algoword{If} $k = 0$ \algoword{then} \algoword{Return} $\sys$
    \item \algoword{Else}
      \begin{enumerate}[{\bf a.}]
        \item $(\anyMat\,,\mat{G}) \assign \algoname{PM-QuoRem}(\modMat,\var^{2^{k-1}\degExp} \sys,2^{k-1}\degExp)$
        \item $\left(\begin{bmatrix} \remMat_{\xpnt 0} \\ \remMat_{\xpnt 1} \end{bmatrix}\right)_{\! 0\le \xpnt<2^{k-1}} \!\!\assign \algoname{RemOfShifts}\!\left(\modMat,\begin{bmatrix} \sys \\ \mat{G} \end{bmatrix},\degExp,k-1\right)$
        \item \algoword{Return} $(\remMat_{\xpnt 0})_{0\le \xpnt <2^{k-1}} \cup (\remMat_{\xpnt 1})_{0\le \xpnt<2^{k-1}}$
      \end{enumerate}
  }
\end{algobox}

\begin{proposition}
  \label{prop:algo:rem_of_shifts}
  \cref{algo:rem_of_shifts} is correct. Assuming that both $2^{k}\nbun\degExp$
  and $\nbeq$ are in $\bigO{\vsdim}$, where $\vsdim = \sumVec{\cdeg{\modMat}}$,
  this algorithm uses $\softO{(2^k\nbun\nbeq^{\expmatmul-2} + k
  \nbeq^{\expmatmul-1}) \vsdim}$ operations in $\field$.
\end{proposition}
\begin{proof}
  The correctness is a consequence of the two properties in
  \cref{lem:quo_rem_properties}.
  Now, if $2^{k}\nbun\degExp$ and $\nbeq$ are in $\bigO{\vsdim}$, the
  assumptions in \cref{prop:algo:quo_rem} about the input parameters for
  \algoname{PM-QuoRem} are always satisfied in recursive calls, since the row
  dimension $\nbun$ is doubled while the exponent $2^k \degExp$ is halved. From
  the same proposition, we deduce the cost bound $\softO{(\sum_{0\le \xpnt\le
  k-1} \lceil 2^r \nbun/\nbeq \rceil ) \nbeq^{\expmatmul-1}\vsdim}$.
\end{proof}

\begin{algobox}
  \algoInfo
  {Residual}
  {algo:residual}

  \dataInfos{Input}{
    \item $\modMat \in \modSpace$ column reduced,
    \item $\sys \in \sysSpace$ such that $\cdeg{\sys}<\cdeg{\modMat}$,
    \item $\relbas \in \relbasSpace$.
  }

  \dataInfo{Output}{
    the remainder $\rem{\relbas\sys}{\modMat}$.
  }

  \algoSteps{
    \item \inlcomment{expand high-degree columns of $\relbas$} \\
      $(\minDeg_i)_{1\le i \le \nbun} \assign \cdeg{\relbas}$ \\
      $\degExp \assign \lceil (\minDeg_1 + \cdots + \minDeg_\nbun) / \nbun \rceil$ \\
      $\quoExp_i \assign \max(1,\lceil \minDeg_i / \degExp \rceil)$ for $1 \le i \le \nbun$  \\
      $\expand{\nbun} \assign \quoExp_1 + \cdots + \quoExp_\nbun$ \\
      $\expand{\relbas} \in \polMatSpace[\expand{\nbun}][\nbun] \assign$
      matrix such that $\relbas = \expand{\relbas} \expandMat$ and $\deg(\expand{\relbas}) \le \degExp$ \\
      \phantom{hfill} \hfill for $\expandMat$ as in~\cref{eqn:expandMat}
    \item \inlcomment{compute $\expand{\sys} = \rem{\expandMat\sys}{\modMat}$} \\
      \algoword{For} $1 \le i \le \nbun$ such that $\quoExp_i=1$ \algoword{do}
      \begin{enumerate}[ ]
        \item $\expand{\sys}_i \in \polMatSpace[\quoExp_i][\nbeq] \assign$ row $i$ of $\sys$
      \end{enumerate}
      \algoword{For} $1 \le k \le \lceil \log_2(\max_i(\quoExp_i)) \rceil$ \algoword{do}
      \begin{enumerate}[ ] 
        \item $(i_1,\ldots,i_\ell) \assign \{ i \in \{1,\ldots,\nbun\} \mid 2^{k-1} < \quoExp_i \le 2^k\}$
        \item $\res \assign $ submatrix of $\sys$ formed by its rows $i_1,\ldots,i_\ell$
        \item $(\remMat_\xpnt)_{0\le \xpnt<2^k} \assign \algoname{RemOfShifts}(\modMat,\res,\degExp,k)$
        \item \algoword{For} $1\le j\le\ell$ \algoword{do}
          \begin{enumerate}[ ]
            \item $\expand{\sys}_{i_j} \in \polMatSpace[\quoExp_{i_j}][\nbeq]
              \assign$ stack the rows $j$ of $(\remMat_\xpnt)_{0\le \xpnt<\quoExp_{i_j}}$
          \end{enumerate}
      \end{enumerate}
      $\expand{\sys} \assign \trsp{\begin{bmatrix} \trsp{\expand{\sys}_1} & \cdots & \trsp{\expand{\sys}_\nbun} \end{bmatrix}} \in \polMatSpace[\expand{\nbun}][\nbeq]$
    \item \inlcomment{left-multiply by the expanded $\relbas$} \\
      $\res \assign \expand{\relbas} \, \expand{\sys}$
    \item \inlcomment{complete the remainder computation} \\
      $(\anyMat\,,\remMat) \assign \algoname{PM-QuoRem}(\modMat,\res,\degExp)$ \\
      \algoword{Return} $\remMat$
  }
\end{algobox}

\begin{proposition}
  \label{prop:algo:residual}
  \cref{algo:residual} is correct. Assuming that all of
  $\sumVec{\cdeg{\relbas}}$, $\nbun$, and $\nbeq$ are in $\bigO{\vsdim}$, where
  $\vsdim = \sumVec{\cdeg{\modMat}}$, this algorithm uses
  $\softO{(\nbun^{\expmatmul-1} + \nbeq^{\expmatmul-1}) \vsdim} $ operations in
  $\field$.
\end{proposition}
\begin{proof}
  Let us consider $\expandMat \in \polMatSpace[\expand{\nbun}][\nbun]$ defined
  as in \cref{eqn:expandMat} from the parameters $\degExp$ and
  $\quoExp_1,\ldots,\quoExp_\nbun$ in Step\,\textbf{1}. We claim that the
  matrix $\expand{\sys}$ computed at Step\,\textbf{2} is equal to
  $\rem{\expandMat\sys}{\modMat}$. Then, having
  $\cdeg{\expand{\relbas}\,\expand{\sys}} < \cdeg{\modMat} +
  (\degExp,\ldots,\degExp)$, the correctness of \algoname{PM-QuoRem} implies
  $\remMat = \rem{\expand{\relbas}\, \expand{\sys}}{\modMat}$, which is
  $\rem{\relbas \sys}{\modMat}$ by \cref{lem:quo_rem_properties}.

  To prove our claim, it is enough to show that, for $1 \le i \le \nbun$, the
  $i$th block $\expand{\sys}_i$ of $\expand{\sys}$ is the matrix formed by
  stacking the remainders involving the row $i$ of $\sys$, that is,
  $(\rem{x^{\xpnt\degExp} \matrow{\sys}{i}}{\modMat})_{0 \le \xpnt <
  \quoExp_i}$. This is clear from the first \algoword{For} loop if $\quoExp_i
  =1$. Otherwise, let $k \in \ZZp$ be such that $2^{k-1} < \quoExp_i \le 2^k$.
  Then, at the $k$th iteration of the second loop, we have $i_j = i$ for some
  $1 \le j\le \ell$. Thus, the correctness of \algoname{RemOfShifts} implies
  that, for $0 \le \xpnt < 2^k$, the row $j$ of $\remMat_\xpnt$ is
  $\rem{x^{\xpnt\degExp} \matrow{\res}{j}}{\modMat} = \rem{x^{\xpnt\degExp}
  \matrow{\sys}{i}}{\modMat}$. Since $2^k \ge \quoExp_i$, this contains the
  wanted remainders and the claim follows.

  Let us show the cost bound, assuming that $\sumVec{\cdeg{\relbas}}$, $\nbun$,
  and $\nbeq$ are in $\bigO{\vsdim}$. Note that this implies $\nbun\degExp \in
  \bigO{\vsdim}$.
  
  We first study the cost of the iteration $k$ of the second loop of
  Step\,\textbf{2}. We have that $2^{k-1} \ell \le \quoExp_1 + \cdots +
  \quoExp_\nbun = \expand{\nbun} \le 2\nbun$, the row dimension of $\res$ is
  $\ell$, and $k \le \lceil\log(\max_i(\quoExp_i))\rceil \in
  \bigO{\log(\nbun)}$. Thus, the call to \algoname{RemOfShifts} costs
  $\softO{(\nbun\nbeq^{\expmatmul-2} + \nbeq^{\expmatmul-1})\vsdim}$ operations
  according to \cref{prop:algo:rem_of_shifts}, and the same cost bound holds
  for the whole Step\,\textbf{2}. Concerning Step\,\textbf{4}, the cost bound
  $\softO{\lceil\nbun/\nbeq\rceil \nbeq^{\expmatmul-1}\vsdim}$ follows
  directly from \cref{prop:algo:quo_rem}.

  The product at Step\,\textbf{3} involves the $\nbun \times \expand{\nbun}$
  matrix $\expand{\relbas}$ whose degree is at most $\degExp$ and the
  $\expand{\nbun} \times \nbeq$ matrix $\expand{\sys}$ such that
  $\cdeg{\expand{\sys}} < \cdeg{\modMat}$; we recall that $\expand{\nbun} \le
  2\nbun$. If $\nbeq \ge \nbun$, we expand the columns of $\expand{\sys}$
  similarly to how $\expand{\relbas}$ was obtained from $\relbas$: this yields
  a $\expand{\nbun} \times (\le 2\nbeq)$ matrix of degree at most
  $\lceil\vsdim/\nbeq\rceil$, whose left-multiplication by $\expand{\relbas}$
  directly yields $\expand{\relbas} \, \expand{\sys}$ by compressing back the
  columns. Thus, this product is done in $\softO{\nbun^{\expmatmul-2} \nbeq
\vsdim}$ operations since both $\degExp$ and $\vsdim/\nbeq$ are in
$\bigO{\vsdim/\nbun}$ when $\nbeq\ge\nbun$. If $\nbun \ge \nbeq$, we do a
similar column expansion of $\expand{\sys}$, yet into a matrix with
$\bigO{\nbun}$ columns and degree $\bigO{\vsdim / \nbun}$; thus, the product
can be performed in $\softO{\nbun^{\expmatmul-1} \vsdim}$ operations in this
case.
\end{proof}

\section{Fast algorithms in specific cases}
\label{sec:building_blocks}

Here, we discuss fast solutions to specific instances of \cref{pbm:relbas}.
This will be important ingredients of our main algorithm for relations modulo
Hermite forms (\cref{algo:relbas}).

\subsection{When the input module is an ideal}
\label{subsec:building_blocks:rk_one}

We first focus on \cref{pbm:relbas} when $\nbeq=1$; this is one of the two base
cases of the recursion in \cref{algo:relbas} (Step\,\textbf{2}). In this case,
the input matrix $\modMat$ is a nonzero polynomial $M\in\polRing$. In other
words, the input module is the ideal $\idealGen{M}$ of $\polRing$, and we are
looking for the $\shifts$-Popov basis for the set of relations between $\nbun$
elements of $\polRing/\idealGen{M}$. A fast algorithm for this task was given
in \cite[Sec.\,2.2]{Neiger16}; precisely, the following result is achieved by
running \cite[Alg.\,2]{Neiger16} on input $\modMat,\sys,\shifts,2\vsdim$.

\begin{proposition}
  \label{prop:algo:rk_one}
  Assuming $\nbeq = 1$ and $\deg(\sys) < \vsdim = \deg(\modMat)$, there is an
  algorithm which solves \cref{pbm:relbas} using
  $\softO{\nbun^{\expmatmul-1}\vsdim}$ operations in $\field$.
\end{proposition}

\subsection{When the $\shifts$-minimal degree is known}
\label{subsec:building_blocks:knowndeg}

Now, we consider \cref{pbm:relbas} with an additional input: the
$\shifts$-minimal degree of $\modRelGen$, which is the column degree of its
$\shifts$-Popov basis. This is motivated by a technique from \cite{JeNeScVi16}
and used in \cref{algo:relbas} to control the degrees of all the bases computed
in the process. Namely, we find this $\shifts$-minimal degree recursively, and
then we compute the $\shifts$-Popov relation basis using this knowledge.

The same question was tackled in \cite[Sec.\,3]{GupSto11} and
\cite[Sec.\,2.1]{Neiger16} for a diagonal matrix $\modMat$. Here, we extend
this to the case of a column reduced $\modMat$, relying in particular on the
fast computation of $\rem{\expandMat \sys}{\modMat}$ designed in
\cref{subsec:division:non_uniform}. We first extend \cite[Lem.\,2.1]{Neiger16}
to this more general setting (\cref{lem:kernel_relbas}), and then we give the
slightly modified version of \cite[Alg.\,1]{Neiger16}
(\cref{algo:knowndeg_relbas}).

\begin{lemma}
  \label{lem:kernel_relbas}
  Let $\modMat \in \modSpace$ be column reduced, let $\sys \in \sysSpace$ be
  such that $\cdeg{\sys} < \cdeg{\modMat}$, let $\shifts\in\shiftSpace$.
  Furthermore, let $\relbas \in \relbasSpace$, and let $\shifts[w] \in
  \shiftSpace[\nbeq]$ be such that $\max(\shifts[w]) \le \min(\shifts)$. Then,
  $\relbas$ is the $\shifts$-Popov relation basis for $\modRelGen$ if and only if
  $[\relbas \;\; \mat{Q}]$ is the $\shifts[u]$-Popov kernel basis of $\trsp{ [
    \trsp{\sys} \;\; \modMat ] }$ for some $\mat{Q} \in
    \polMatSpace[\nbun][\nbeq]$ and $\shifts[u] = (\shifts,\shifts[w]) \in
    \shiftSpace[\nbun+\nbeq]$. In this case, $\deg(\mat{Q}) < \deg(\relbas)$
    and $[ \relbas \;\; \mat{Q} ]$ has $\shifts[u]$-pivot index
    $(1,2,\ldots,\nbun)$.
\end{lemma}
\begin{proof}
  Let $\kerbas = \trsp{[ \trsp{\sys} \;\; \modMat ]}$. It is easily verified
  that $\relbas$ is a relation basis for $\modRelGen$ if and only if there is
  some $\mat{Q}\in \polMatSpace[\nbun][\nbeq]$ such that $[\relbas \;\;
  \mat{Q}]$ is a kernel basis of $\kerbas$.

  Then, for any matrix $[\mat{P} \;\; \mat{Q}] \in
  \polMatSpace[\nbun][(\nbun+\nbeq)]$ in the kernel of $\kerbas$, we have
  $\mat{P} \sys = -\mat{Q} \modMat$ and therefore \cref{cor:quotient_degree}
  shows that $\rdeg{\mat{Q}} < \rdeg{\mat{P}}$; since $\max(\shifts[w]) \le
  \min(\shifts)$, this implies $\rdeg[{\shifts[w]}]{\mat{Q}} <
  \rdeg[\shifts]{\mat{P}}$. Thus, we have
  $\leadingMat[{\shifts[u]}]{[\mat{P}\;\;\mat{Q}]} =
  [\leadingMat[\shifts]{\mat{P}} \;\; \matz]$, and therefore $\mat{P}$ is in
  $\shifts$-Popov form if and only if $[\mat{P} \;\; \mat{Q}]$ is in
  $\shifts[u]$-Popov form with $\shifts[u]$-pivot index $(1,\ldots,\nbun)$.
\end{proof}

\begin{algobox}
  \algoInfo
  {KnownDegreeRelations}
  {algo:knowndeg_relbas}

  \dataInfos{Input}{
    \item $\modMat \in \modSpace$ column reduced,
    \item $\sys \in \sysSpace$ such that $\cdeg{\sys} < \cdeg{\modMat}$,
    \item $\shifts \in \shiftSpace$,
    \item $\minDegs = (\minDeg_1,\ldots,\minDeg_\nbun)$ the $\shifts$-minimal
      degree of $\modRelGen$.
  }

  \dataInfo{Output}{
    the $\shifts$-Popov relation basis for $\modRelGen$.
  }

  \algoSteps{
    \item \inlcomment{define partial linearization parameters} \\
      $\degExp \assign \lceil (\minDeg_1 + \cdots + \minDeg_\nbun) / \nbun \rceil$, \\
      $\quoExp_i \assign \max(1,\lceil \minDeg_i / \degExp \rceil)$ for $1 \le i \le \nbun$,  \\
      $\expand{\nbun} \assign \quoExp_1 + \cdots + \quoExp_\nbun$, \\
      $\expand{\minDegs} \assign$ tuple as in~\cref{eqn:expandMinDegs}
    \item \inlcomment{for $\expandMat$ as in \cref{eqn:expandMat}, compute $\expand{\sys} = \rem{\expandMat\sys}{\modMat}$} \\
      $\expand{\sys} \assign$ follow Step\,\textbf{2} of \cref{algo:residual} (\algoname{Residual})
    \item \inlcomment{compute the kernel basis} \\
      $\shifts[u] \assign (-\expand{\minDegs},-\degExp,\ldots,-\degExp)
      \in \shiftSpace[\expand{\nbun}+\nbeq]$ \\
      $\boldsymbol{\tau} \assign (\cdeg{\matcol{\modMat}{j}} + \degExp+1)_{1 \le j \le \nbeq}$ \\
      $\expand{\relbas} \assign$ $\shifts[u]$-Popov approximant basis for
      $\begin{bmatrix} \expand{\sys} \\ \modMat\end{bmatrix}$ and orders $\boldsymbol{\tau}$
    \item \inlcomment{retrieve the relation basis} \\
      $\mat{P}$ $\assign$ the principal $\expand{\nbun} \times \expand{\nbun}$ submatrix of $\expand{\relbas}$ \\
      \algoword{Return} the submatrix of $\relbas \expandMat$ formed by the
      rows at indices $\quoExp_1+\cdots+\quoExp_i$ for $1\le i\le \nbun$
  }
\end{algobox}

\begin{proposition}
  \label{prop:algo:knowndeg}
  \cref{algo:knowndeg_relbas} is correct, and assuming that $\nbun$ and $\nbeq$
  are in $\bigO{\vsdim}$, where $\vsdim = \sumVec{\cdeg{\modMat}}$, it uses
  $\softO{\nbun^{\expmatmul-1} \vsdim + \nbeq^\expmatmul \vsdim/\nbun}$ operations in
  $\field$.
\end{proposition}
\begin{proof}
  The correctness follows from the material in \cite[Sec.\,2.1]{Neiger16} and
  \cite[Sec.\,4]{JeNeScVi16}. Concerning the cost bound, we first note that we
  have $\minDeg_1+\cdots+\minDeg_\nbun \le \vsdim$ according to
  \cref{lem:degdetbound}. Thus, the cost analysis in \cref{prop:algo:residual}
  shows that Step\,\textbf{2} uses $\softO{(\nbun\nbeq^{\expmatmul-2} +
  \nbeq^{\expmatmul-1})\vsdim}$ operations. \cite[Thm.\,1.4]{JeNeScVi16} states
  that the approximant basis computation at Step\,\textbf{3} uses
  $\softO{(\nbun+\nbeq)^{\expmatmul-1} (1 + \nbeq/\nbun) \vsdim}$ operations,
  since the row dimension of the input matrix is $\expand{\nbun}+\nbeq \le
  2\nbun+\nbeq$ and the sum of the orders is $\sumVec{\boldsymbol{\tau}} =
  \sumVec{\cdeg{\modMat}} + \nbeq(\degExp+1) \le (1 + \nbeq/\nbun)\vsdim$.
\end{proof}

\subsection{Solution based on fast linear algebra}
\label{subsec:linalg}

Here, we detail how previous work can be used to handle a base case of the
recursion in \cref{algo:relbas} (Step\,\textbf{1}): when the vector space
dimension $\deg(\det(\modMat))$ of the input module is small compared to the
number $\nbun$ of input elements. Then, we rely on an interpretation of
\cref{pbm:relbas} as a question of dense linear algebra over $\field$, which is
solved efficiently by \cite[Alg.\,9]{JeNeScVi17}.  This yields the following
result.

\begin{proposition}
  \label{prop:relbas_linalg}
  Assuming that $\modMat$ is in shifted Popov form, and that $\cdeg{\sys} <
  \cdeg{\modMat}$, there is an algorithm which solves \cref{pbm:relbas} using
  $\softO{\vsdim^{\expmatmul} \lceil \nbun/\vsdim \rceil}$ operations in
  $\field$, where $\vsdim = \deg(\det(\modMat))$.
\end{proposition}

This cost bound is $\softO{\vsdim^{\expmatmul-1} \nbun} \subseteq
\softO{\nbun^{\expmatmul-1} \vsdim}$ when $\vsdim \in \bigO{\nbun}$. To see why
relying on fast linear algebra is sufficient to obtain a fast algorithm when
$\vsdim \in \bigO{\nbun}$, we note that this implies that the average column
degree of the $\shifts$-Popov relation basis $\relbas$ is 
\[
  \sumVec{\cdeg{\relbas}}/\nbun = \deg(\det(\relbas))/\nbun \le
  \vsdim / \nbun \in
\bigO{1}.
\]

For example, if $\vsdim \le \nbun$, most entries in this basis have degree $0$:
we are essentially dealing with matrices over $\field$. On the other hand, when
$\nbun \in \bigO{\vsdim}$, this approach based on linear algebra uses
$\softO{\vsdim^\expmatmul}$ operations, which largely exceeds our target cost.

We now describe how to translate our problem into the $\field$-linear algebra
framework in \cite{JeNeScVi17}. Let $\module$ denote the row space of
$\modMat$; we assume that $\modMat$ has no identity column. In order to
compute in the quotient $\polRing^\nbeq / \module$, which has finite dimension
$\vsdim$, it is customary to make use of the \emph{multiplication matrix} of
$\var$ with respect to a given monomial basis.  Here, since the basis $\modMat$
of $\module$ is in shifted Popov form with column degree
$(\order_1,\ldots,\order_\nbeq) \in \ZZp^\nbeq$, \cref{lem:vsdim_colredbasis}
suggests to use the monomial basis
\[
  \{(\var^{i},0,\ldots,0), 0 \le i < \order_1\}
  \cup \cdots \cup
  \{(0,\ldots,0,\var^{i}), 0 \le i < \order_n\}.
\]
Above, we have represented an element in $\polRing^\nbeq / \module$ by a
polynomial vector $\row{f} \in \polMatSpace[1][\nbeq]$ such that
$\cdeg{\row{f}} < (\order_1,\ldots,\order_\nbeq)$. In the linear algebra
viewpoint, we rather represent it by a constant vector $\evRow \in
\matSpace[1][\vsdim]$, which is formed by the concatenations of the coefficient
vectors of the entries of $\row{f}$. Applying this to each row of the input
matrix $\sys$ yields a constant matrix $\evMat \in \evMatSpace$, which is
another representation of the same $\nbun$ elements in the quotient.

Besides, the multiplication matrix $\mulmats \in \mulmatSpace$ is the matrix
such that $\evRow \mulmats \in \matSpace[1][\vsdim]$ corresponds to the
remainder in the division of $\var \row{f}$ by $\modMat$. Since the basis
$\modMat$ is in shifted Popov form, the computation of $\mulmats$ is
straightforward. Indeed, writing $\modMat =
\diag{\var^{\order_1},\ldots,\var^{\order_\nbeq}} - \mat{A}$ where $\mat{A} \in
\modSpace$ is such that $\cdeg{\mat{A}} < (\order_1,\ldots,\order_\nbeq)$, then
\begin{itemize}
  \item the row $\order_1 + \cdots + \order_{i-1} + j$ of $\mulmats$ is the
    unit vector with $1$ at index $\order_1+\cdots+\order_{i-1} + j+1$, for $1
    \le j < \order_i$ and $1 \le i \le \nbeq$,
  \item the row $\order_1 + \cdots + \order_{i}$ of $\mulmats$ is the
    concatenation of the coefficient vectors of the row $i$ of $\mat{A}$, for
    $1 \le i \le \nbeq$.
\end{itemize}
That is, writing $\mat{A} = [a_{ij}]_{1 \le i,j \le \nbeq}$ and denoting by
$\{a_{ij}^{(k)}, 0 \le k < \order_j\}$ the coefficients of $a_{ij}$, the
multiplication matrix $\mulmats \in \matSpace[\vsdim]$ is
\[
    \begin{bmatrix}
      & 1                                    \\
      &   & \ddots                           \\
      &   &        & 1                       \\
      a_{11}^{(0)} & a_{11}^{(1)} & \;\cdots\; & a_{11}^{(\order_1-1)} & \cdots & a_{1\nbeq}^{(0)} & a_{1\nbeq}^{(1)} & \cdots & a_{1\nbeq}^{(\order_\nbeq-1)} \\[0.2cm]
      &   &        &   & \;\ddots\; \\[0.2cm]
      &   &        &   &            &   & 1               \\
      &   &        &   &            &   &   & \ddots      \\
      &   &        &   &            &   &   &        & 1  \\
      a_{\nbeq 1}^{(0)} & a_{\nbeq 1}^{(1)} & \;\cdots\; & a_{\nbeq 1}^{(\order_1-1)} & \cdots & a_{\nbeq\nbeq}^{(0)} & a_{\nbeq\nbeq}^{(1)} & \cdots & a_{\nbeq\nbeq}^{(\order_\nbeq-1)}
    \end{bmatrix}
  .
\]

\section{Relations modulo Hermite forms}
\label{sec:relbas}

In this section, we give a fast algorithm for solving \cref{pbm:relbas} when
$\modMat$ is in Hermite form; this matrix is denoted by $\modHer$ in what
follows. The cost bound is given under the assumption that $\modHer$ has no
identity column; how to reduce to this case by discarding columns of $\modHer$
and $\sys$ was discussed in \cref{cor:cleaning}. We recall that
Steps\,\textbf{1}, \textbf{2}, and \textbf{3.i} have been discussed in
\cref{sec:building_blocks}.

\begin{proposition}
  \label{prop:algo:relbas}
  \cref{algo:relbas} is correct and, assuming the entries $\cdeg{\modHer}$ are
  positive, it uses $\softO{\nbun^{\expmatmul-1} \vsdim + \nbeq^\expmatmul
\vsdim / \nbun}$ operations in $\field$, where $\vsdim =
\sumVec{\cdeg{\modHer}} = \deg(\det(\modHer))$. 
\end{proposition}
\begin{proof}
  Following the recursion in the algorithm, our proof is by induction on
  $\nbeq$, with two base cases (Steps\,\textbf{1} and\,\textbf{2}).

  The correctness and the cost bound for Step\,\textbf{1} follows from the
  discussion in \cref{subsec:linalg}, as summarized in
  \cref{prop:relbas_linalg}. From \cref{subsec:building_blocks:rk_one},
  Step\,\textbf{2} correctly computes the $\shifts$-Popov relation basis and
  uses $\softO{\nbun^{\expmatmul-1} \vsdim}$ operations in $\field$.

  Now, we focus on the correctness of Step\,\textbf{3}, assuming that the two
  recursive calls at Steps\,\textbf{3.d} and\,\textbf{3.g} correctly compute
  the shifted Popov relation bases. Since \algoname{KnownDegreeRelations} is
  correct, it is enough to prove that the $\shifts$-minimal degree of $\modRel$
  is $\minDegs_1 + \minDegs_2$; for this, we will show that $\relbas_2
  \relbas_1$ is a relation basis for $\modRel$ whose $\shifts$-Popov form has
  column degree $\minDegs_1 + \minDegs_2$.

  From \cref{thm:basis_transitivity}, $\relbas_2 \relbas_1$ is a relation basis
  for $\modRel$. Furthermore, the fact that the $\shifts$-Popov form of
  $\relbas_2 \relbas_1$ has column degree $\minDegs_1 + \minDegs_2$ follows
  from \cite[Sec.\,3]{JeNeScVi16}, since $\relbas_1$ is in $\shifts$-Popov form
  and $\relbas_2$ is in $\shifts[t]$-Popov form, where $\shifts[t] = \shifts +
  \minDegs_1 = \rdeg[\shifts]{\relbas_1}$.

  Concerning the cost of Step\,\textbf{3}, we remark that $\nbun < \vsdim$,
  that $\nbeq \le \vsdim$ is ensured by $\cdeg{\modHer} > \unishift$, and that
  $\minDegs_1+\minDegs_2=\deg(\det(\relbas_2\relbas_1)) \le \vsdim$ according
  to \cref{lem:degdetbound}.  Furthermore, there are two recursive calls with
  dimension about $\nbeq/2$, and with $\modHer_1$ and $\modHer_2$ that are in
  Hermite form and have determinant degrees $\vsdim_1 = \deg(\det(\modHer_1))$
  and $\vsdim_2 = \deg(\det(\modHer_2))$ such that $\vsdim = \vsdim_1 +
  \vsdim_2$. Besides, the entries of both $\cdeg{\modHer_1}$ and
  $\cdeg{\modHer_2}$ are all positive.
  
  In particular, the assumptions on the parameters in
  \cref{prop:algo:residual,prop:algo:knowndeg}, concerning the computation of
  the residual at Step\,\textbf{3.f} and of the relation basis when the degrees
  are known at Step\,\textbf{3.i}, are satisfied. Thus, these steps use
  $\softO{(\nbun^{\expmatmul-1} + \nbeq^{\expmatmul-1}) \vsdim} $ and
  $\softO{\nbun^{\expmatmul-1} \vsdim + \nbeq^\expmatmul \vsdim/\nbun}$
  operations, respectively. The announced cost bound follows.
\end{proof}

\begin{algobox}
  \algoInfo
  {RelationsModHermite}
  {algo:relbas}

  \dataInfos{Input}{
    \item matrix $\modHer \in \modSpace$ in Hermite form,
    \item matrix $\sys \in \sysSpace$ such that $\cdeg{\sys} < \cdeg{\modHer}$,
    \item shift $\shifts \in \shiftSpace$.
  }

  \dataInfo{Output}{
    the $\shifts$-Popov relation basis for $\modRel$.\medskip
  }

  \algoSteps{
    \item \algoword{If} $\vsdim = \sumVec{\cdeg{\modHer}} \le \nbun$:
      \begin{enumerate}[{\bf a.}]
        \item build $\mulmats \in \matSpace[\vsdim]$ from $\modHer$ as in \cref{subsec:linalg}
        \item build $\evMat \in \matSpace[\nbun][\vsdim]$ from $\sys$ as in \cref{subsec:linalg}
        \item $\relbas \assign$ \cite[Alg.\,9]{JeNeScVi17} on input $(\evMat,\mulmats,\shifts,2^{\lceil\log_2(\vsdim)\rceil})$
        \item \algoword{Return} $\relbas$
      \end{enumerate}
    \item \algoword{Else if} $\nbeq = 1$ \algoword{then}
      \begin{enumerate}[{\bf a.}]
        \item $\relbas \assign$ \cite[Alg.\,2]{Neiger16} on input $(\modHer,\sys,\shifts,2\vsdim)$
        \item \algoword{Return} $\relbas$
      \end{enumerate}
    \item \algoword{Else}:
      \begin{enumerate}[{\bf a.}]
        \item $\nbeq_1 \assign \lfloor \nbeq / 2 \rfloor$; $\nbeq_2 \assign \lceil \nbeq / 2 \rceil$
        \item $\modHer_1$ and $\modHer_2$ $\assign$
          the $\nbeq_1 \times \nbeq_1$ leading and $\nbeq_2 \times
          \nbeq_2$ trailing principal submatrices of $\modHer$
        \item $\sys_1 \assign$ first $\nbeq_1$ columns of $\sys$
        \item $\relbas_1 \assign \algoname{RelationsModHermite}(\modHer_1,\sys_1,\shifts)$
        \item $\minDegs_1 \assign$ diagonal degrees of $\relbas_1$
        \item $\res \assign$ last $\nbeq_2$ columns of $\algoname{Residual}(\modHer,\relbas_1,\sys)$
        \item $\relbas_2 \assign \algoname{RelationsModHermite}(\modHer_2,\res,\shifts+\minDegs_1)$
        \item $\minDegs_2 \assign$ diagonal degrees of $\relbas_2$
        \item \algoword{Return} $\algoname{KnownDegreeRelations}(\modHer,\sys,\shifts,\minDegs_1+\minDegs_2)$ 
      \end{enumerate}
  }
\end{algobox}

\begin{acks}
  The authors thank Claude-Pierre Jeannerod for interesting discussions, Arne
  Storjohann for his helpful comments on high-order lifting, and the reviewers
  whose remarks helped to prepare the final version of this paper. The research
  leading to these results has received funding from the People Programme
  (Marie Curie Actions) of the European Union's Seventh Framework Programme
  (FP7/2007-2013) under REA grant agreement number 609405 (COFUNDPostdocDTU).
  Vu Thi Xuan acknowledges financial support provided by the scholarship
  \emph{Explora Doc} from \emph{R\'egion Rh\^one-Alpes, France}, and by the
  LABEX MILYON (ANR-10-LABX-0070) of Universit\'e de Lyon, within the program
  \emph{Investissements d'Avenir} (ANR-11-IDEX-0007) operated by the French
  National Research Agency.
\end{acks}

\bibliographystyle{ACM-Reference-Format}

\end{document}